\documentclass{tlp}

\usepackage{amsmath,amsfonts,amssymb}
\usepackage{xypic}
\usepackage[utf8]{inputenc}
\usepackage{maybemath}
\usepackage{multirow}
\newcommand{\comment}[1]{}
\usepackage{url}

\newcommand{\Conf}{\mathit{Conf}}

\newcommand{\code}[1]{{\tt #1}}

\newcommand{\HLnofunctors}{\ensuremath{C}}
\newcommand{\HLfunctors}{\ensuremath{F}}

\newcounter{mycounter}
\newtheorem{mylemma}[mycounter]{Lemma}

\newtheorem{theorem}{Theorem}[section]

\newtheorem{corollary}[theorem]{Corollary}

\newtheorem{lemma}[theorem]{Lemma}
\newtheorem{definition}[theorem]{Definition}

\newtheorem{example}[theorem]{Example}

\title{Decidability properties for fragments of CHR}

\author[M. Gabbrielli et.al.]{
        MAURIZIO GABBRIELLI\\
        Dipartimento di Scienze dell’Informazione and Lab. Focus INRIA, Universit\`a di Bologna\\
        \email{gabbri@cs.unibo.it}\\
        \and
        JACOPO MAURO\\
        Dipartimento di Scienze dell'Informazione, Universit\`a  di Bologna\\
        \email{jmauro@cs.unibo.it}\\
        \and
        MARIA CHIARA MEO\\
        Dipartimento di Scienze, Universit\`a di Chieti Pescara
\\
\email{cmeo@unich.it}        \and
        JON SNEYERS\\
        Departement Computerwetenschappen, K.U.Leuven \\
        \email{jon.sneyers@cs.kuleuven.be}
}
\comment{
\author{Maurizio Gabbrielli \inst{1} \and Jacopo Mauro \inst{1} \and Maria Chiara Meo \inst{2} \and Jon Sneyers \inst{3}}
\institute{
Dipartimento di Scienze dell'Informazione, Universit\`a  di Bologna \\
\texttt{gabbri | jmauro@cs.unibo.it}
\and
Dipartimento di Scienze, Universit\`a di Chieti Pescara}
\\
\texttt{cmeo@unich.it}
\and
Departement Computerwetenschappen, K.U.Leuven \\
\texttt{jon.sneyers@cs.kuleuven.be}
}

\begin{document}

\maketitle

\begin{abstract}
We study the decidability of termination for two CHR dialects which, similarly to the Datalog like
languages, are defined by using a signature which does not allow function symbols (of arity $>0$). Both languages
allow the use of the $=$ built-in in the body of rules, thus are built on a host language that supports unification.
However each imposes one further restriction. The first CHR dialect allows only {\em range-restricted} rules,
that is, it does not allow the use of variables in the body or in the guard of a rule if they do not appear in the head.
We show that the existence of an infinite computation is decidable for this dialect. The second dialect instead limits 
the number of atoms in the head of rules to one. We prove that in this case, the existence of a terminating
computation is decidable. These results show that both dialects are strictly less expressive\footnote{As we clarify later, ``less expressive'' here means that there exists no termination preserving encoding of Turing machines in the considered language.}  than Turing
Machines. It is worth noting that the language (without function symbols) without these restrictions is as expressive as Turing Machines.
\end{abstract}


\section{Introduction}

Constraint Handling Rules (CHR) \cite{fru_chr_overview_jlp98,fru_chr_2009} is a
declarative general-purpose language. A CHR program consists of a set of multi-headed guarded
(simplification, propagation and simpagation) rules which allow one
to rewrite constraints into simpler ones until a solved form is
reached. The language is parametric w.r.t. an underlying constraint theory $\mathcal{CT}$
which defines %
basic built-in constraints. For a recent survey on the language
see \citeN{newsurvey}.

In the last few years, several papers have %
investigated
the expressivity of CHR,
however very few decidability results for fragments of CHR have been obtained. Three main aspects
affect the computational power of CHR: the number of atoms allowed in the heads, the nature of the
underlying signature on which programs are defined, and the constraint theory. %
The latter two aspects are often referred to %
as
the ``host language''
since they identify the language on which a CHR system is built.
Some results in \cite{DBLP:conf/sofsem/GiustoGM09} indicate that %
restricting to single-headed
rules decreases the computational power of CHR. However, these results consider Turing complete
fragments of CHR, hence they do not establish any decidability result. Indeed, single-headed CHR
is Turing-complete  \cite{DBLP:conf/sofsem/GiustoGM09}, provided that the host language allows
functors and %
unification. On the other hand, when allowing multiple heads, even
restricting to a host language which allows only constants does not allow to obtain any
decidability property, since even with this limitation CHR is Turing complete
\cite{DBLP:conf/iclp/Sneyers08,DBLP:conf/sofsem/GiustoGM09}. The only (implicit) decidability results
concern propositional CHR, where all constraints have arity 0, and %
CHR without functors and without unification, since these languages can be
translated to (colored) Petri Nets \cite{betz_petri_nets_chr07}
--- see also %
Section \ref{sec:conclusions}.

Given this situation, when looking for decidable properties it is natural to consider
further restrictions of the above mentioned CHR language which allows the only built-in $=$ (interpreted in the usual way as equality on the Herbrand universe) and which, similarly to Datalog, is defined over a signature which contains no function symbol of arity $> 0$. We denote such a language by CHR($\HLnofunctors$).

In this paper we provide two decidability results for two fragments of CHR($\HLnofunctors$). The first fragment allows {\em range-restricted} rules only, that is, it does not allow the use of a variable in the body or in the guard if it  does not appear in the head. We show, using the theory of well-structured transition systems \cite{FS01,DBLP:conf/lics/AbdullaCJT96},  that in this case the existence of an infinite computation is decidable.
The second fragment that we consider is single-headed CHR($\HLnofunctors$), denoted by CHR$_1$($\HLnofunctors$).   We prove that, for this language, the existence of a terminating computation is decidable. In this case we provide a direct proof, since no reduction to Petri Nets can be used (the language introduces an infinite states system) and well-structured transition system can not be used (they do not allow to prove this kind of decidability properties).

These results show that both CHR fragments are strictly less expressive than Turing Machines. As previously mentioned, CHR($\HLnofunctors$) is as expressive as Turing Machines. So these results obviously imply that both restrictions lower the expressive power of CHR($\HLnofunctors$).

\section{Syntax and semantics}
In this section we give an overview of CHR syntax and its operational semantics
following \cite{fru_chr_overview_jlp98,duck_stuck_garc_holz_refined_op_sem_iclp04}. A constraint $c(t_1,\dots,t_n)$ is an atomic formula constructed on
a given signature $\Sigma$ in the usual way.  There are two types of constraints: built-in constraints
(predefined) that are handled by an existing solver and CHR constraints (user-defined) which are defined
by a CHR program. Therefore we assume that the signature $\Sigma$ contains two disjoint sets of predicate
symbols for built-in and  CHR constraints. For built-in constraints we assume that a first order
decidable theory $\mathcal{CT}$ is given which describes their meaning. Often the terminology
``host language'' is used to indicate the language consisting of the built-in predicates,
because indeed often CHR is implemented on top of such
an existing host language.

To distinguish between different occurrences of syntactically equal constraints, CHR constraints are extended with a unique identifier. An identified CHR constraint is denoted by $c \# i$ with $c$ a CHR constraint and $i$ the identifier. We write $\code{chr}(c\#i) = c$ and $\code{id}(c\#i) = i$, possibly extended to sets and sequences of identified CHR constraints in the obvious way.

\label{subsection:CHR_program}
A CHR program is defined as a sequence of three kinds of rules: simplification, propagation and simpagation rules. Intuitively, simplification rewrites constraints into simpler ones, propagation adds new constraints which are logically redundant but may trigger further simplifications,
and simpagation combines in one rule the effects of both propagation and simplification rules.
For simplicity we consider simplification and propagation rules as special cases of a simpagation rule. The general form of a simpagation rule is:
$${\it r}\ @ \  H^k \ \backslash \ H^h \ \Longleftrightarrow g \mid B$$
where ${\it r }$ is a
unique identifier of a rule, $H^k$ and $H^h$ (the heads) are multi-sets of CHR constraints,
$g$ (the guard) is a %
conjunction of built-in
constraints and $B$ is a %
multi-set of (built-in and
user-defined) constraints. If $H^k$ is empty then the rule is a simplification rule. If $H^h$ is empty then the rule is a propagation rule. At least one of $H^k$ and $H^h$ must be non-empty.
When the guard $g$ is empty or $true$ we omit  $g  \,|$.
The names of rules are omitted when not needed. For a simplification rule we omit $H^k \backslash$ while we write a propagation rule as $H^k \Longrightarrow g \mid B$.
A CHR {\em goal} is a multi-set of (both user-defined and built-in) constraints.

We also use the following notation: $\exists_{V} \phi$, where $V$ is a set of variables, denotes the
existential closure of a formula $\phi$ w.r.t. the variables in $V$, while $\exists_{-V} \phi$ denotes the existential closure of a formula $\phi$ with the exception of the variables in $V$ which remain unquantified. $Fv(\phi)$ denotes the free variables appearing in $\phi$ and $t \sigma$ the application of a substitution $\sigma$ to a syntactic object $t$.

\medskip
{\bf CHR dialects.}\label{secdialectes} As mentioned before, the computational power of CHR depends on
several aspects, including the number of atoms allowed in the heads, %
the underlying
signature $\Sigma$ on which programs are defined, and the constraint theory $\mathcal{CT}$,
defining the %
built-ins.
We use the notation CHR($X$), where the parameter $X$ indicates the
signature and the constraint theory (in other words, the host language).

More precisely, the language under consideration in this paper is  CHR($\HLnofunctors$) and has been defined in the introduction.
We will also use the notation CHR($P$) to denote {\em propositional} CHR, that is the language where all
constraints have arity zero. This corresponds to consider a trivial host language without any data type.
Finally CHR($\HLfunctors$) indicates the (usual) CHR language which allows functor symbols and the $=$ built-in.
Thus in this case the host language allows arbitrary Herbrand terms and supports unification among them.

The number of atoms in the heads also affects the expressive power of the language.
We use the notation CHR$_1$, possibly combined with the notation above, to denote
{\em single-headed} CHR, where heads of rules contain one atom.

\medskip
{\bf Operational semantics of CHR.}
\label{subsection:traditional_operational_semantics}
We consider the theoretical operational semantics, denoted by $\omega_t$ and the abstract semantics,
denoted by $\omega_o$. The semantics  $\omega_t$ is given by \citeN{duck_stuck_garc_holz_refined_op_sem_iclp04}
as a state transition system $T= (\Conf, \stackrel{\omega_t}{\rightarrow_P})$ where configurations in $\Conf$
are tuples of the form $\langle G,S,B,T \rangle_n$, where
$G$ is the goal (a multi-set of constraints that remain to be solved), $S$ is the CHR store (a set of identified CHR constraints), $B$ is the built-in store (a conjunction of built-in constraints), $T$ is the propagation history (a sequence of identifiers used to store the rule instances fired) and $n$ is the next free identifier (it is used to identify new CHR constraints).
The transitions of $\omega_t$ are shown in Table \ref{table:traditional_semantics}.

Given a program $P$, the transition relation $\stackrel{\omega_t}{\rightarrow_P}
\subseteq \Conf \times \Conf$ is the least relation
satisfying the rules  in Table \ref{table:traditional_semantics}. The {\bf Solve} transition allows to update the constraint store by taking into account a built-in constraint contained in the goal. The {\bf Introduce} transition is used to move a user-defined constraint from the goal
to the CHR constraint store, where it can be handled by applying CHR rules.  The {\bf Apply} transition allows to rewrite user-defined constraints (which are in the CHR constraint store)
using rules from the program.
The {\bf Apply} transition is applicable when the current built-in store ($B$) %
entails
the guard of the rule ($g$).

\begin{table}
\begin{center}
\begin{description}
 \item[Solve]
  $\langle \{c\} \uplus G, S, B, T\rangle_n \stackrel{\omega_t}{\rightarrow_P}
  \langle G, S,c \wedge B, T\rangle_n$ where $c$ is a built-in constraint
 \item[Introduce]
  $\langle \{c\} \uplus G, S, B, T\rangle_n \stackrel{\omega_t}{\rightarrow_P}
  \langle G, \{c\#n\} \cup S,B, T\rangle_{n+1}$ where $c$ is a CHR constraint
 \item[Apply]
  $\langle G, H_1 \cup H_2 \cup S, B, T\rangle_n \stackrel{\omega_t}{\rightarrow_P}
  \langle C \uplus G, H_1 \cup S, \theta \wedge B, T \cup \{t\} \rangle_n$ where $P$ contains a (renamed apart) rule
  ${\it r}\ @ H'_1 \backslash H'_2 \Longleftrightarrow g \mid C$
  and there exists a matching substitution $\theta$ s.t. $\code{chr}(H_1) = H'_1 \theta$, $\code{chr}(H_2) = H'_2 \theta$, $\mathcal{CT} \models B \rightarrow \exists_{-Fv(B)}(\theta \wedge g)$
  \\ and $t = \code{id}(H_1) \mathrel+\joinrel\mathrel+ \code{id}(H_2) \mathrel+\joinrel\mathrel+ [r] \notin T$
\end{description}

\end{center}
\caption{Transitions of $\omega_t$}
\label{table:traditional_semantics}
\end{table}

An {\em initial configuration} has the form
$\langle G,\emptyset,true, \emptyset \rangle_1$
while a  {\em final configuration} has either the form
$\langle G,S,false, T \rangle_k$
when it is {\em failed},
or  the form
$\langle \emptyset,S,B, T \rangle_k$ when it is successfully terminated because there are
no applicable rules.
A computation is called {\em terminating} if it ends in a final configuration, {\em infinite} otherwise.

\label{subsection:original_operational_semantics}

The first CHR operational semantics defined in  \cite{fru_chr_overview_jlp98} differs from the traditional
semantics $\omega_t$. Indeed this original, so called,  abstract semantics denoted by $\omega_o$, allows
the firing of a propagation rule an infinite number of times. For this reason $\omega_o$ can be seen as
the abstraction of the traditional semantics where the propagation history is not considered.
It is identical to $\omega_t$, except that configurations are
of the form $\langle G,S,B \rangle_n$ (they do not contain a propagation history)
and the {\bf Apply} transition does not have the last condition that $t \not\in T$.

\comment{
An {\em initial configuration} has the form
$\langle G,\emptyset,true \rangle_1$
while a  {\em final configuration} has either the form
$\langle G,S,false \rangle_k$
when it is {\em failed},
or  the form
$\langle \emptyset,S,B\rangle_k$ when it is successfully terminated because there are
no applicable rules.
}

\comment{
\begin{table}
\begin{center}
\begin{description}
 \item[Solve]
  $\langle \{c\} \uplus G, S, B\rangle_n \stackrel{\omega_o}{\rightarrow_P}
  \langle G, S,c \wedge B\rangle_n$ where $c$ is a built-in constraint
 \item[Introduce]
  $\langle \{c\} \uplus G, S, B\rangle_n \stackrel{\omega_o}{\rightarrow_P}
  \langle G, \{c\#n\} \cup S,B\rangle_{n+1}$ where $c$ is a CHR constraint
 \item[Apply]
  $\langle G, H_1 \cup H_2 \cup S, B\rangle_n \stackrel{\omega_o}{\rightarrow_P}
  \langle C \uplus G, H_1 \cup S, \theta \wedge B\rangle_n$ where $P$ contains a (renamed apart) rule
  ${\it r}\ @ H'_1 \backslash H'_2 \Longleftrightarrow g \mid C$
  and there exists a matching substitution $\theta$ s.t. $\code{chr}(H_1) = \theta H'_1$, $\code{chr}(H_2) = \theta H'_2$, $\mathcal{CT} \models B \rightarrow \exists_{-Fv(B)}(\theta \wedge g)$
\end{description}

\end{center}
\caption{Transitions of $\omega_o$}
\label{table:original_semantics}
\end{table}
}

\section{Range-restricted CHR($\HLnofunctors$)}\label{sec:rr}

In this section we consider the (multi-headed) range-restricted CHR($\HLnofunctors$) language described in the introduction. We call a CHR rule range-restricted if all the variables which appear in the body and in the guard appear also in the head of a rule. More formally, if $Var(X)$ denotes the variables used in $X$, the rule ${\it r}\ @ H^k \backslash H^h \Longleftrightarrow g \mid B$ is range-restricted if $Var(B)\cup Var(g) \subseteq Var (H^k \backslash H^h)$ holds. A CHR language is called range-restricted if it allows range-restricted rules only.

We prove that in range-restricted CHR($\HLnofunctors$) the existence of an infinite computation is a decidable
property when considering the $\omega_o$ semantics. This shows that  this language is less expressive than Turing Machines and
than CHR($\HLnofunctors$). Our result is based on the theory of well-structured transition systems (WSTS) and we refer to \cite{FS01,DBLP:conf/lics/AbdullaCJT96} for this theory. Here we only provide the basic definitions on WSTS, taken from \cite{FS01}.

Recall that a {\em quasi-order} (or, equivalently, preorder) is a reflexive and transitive relation. A {\em well-quasi-order} (wqo) is defined as a quasi-order $\leq$ over a set $X$
such that, for any infinite sequence $x_0, x_1, x_2, \ldots$
in $X$, there exist indexes $i<j$ such that $x_i\leq x_j$.

A {\em transition system} is defined as usual, namely it is a structure $TS = (S,\rightarrow)$, where $S$ is a set of {\em states} and $\rightarrow\subseteq S\times S$
is a set of {\em transitions}.  We define $Succ(s)$ as the set $\{s'\in S\mid s\rightarrow s'\}$ of immediate successors of $s$.
We say that  $TS$ is {\em finitely branching} if, for each $s \in S$,  $Succ(s)$ is finite.
Hence we have the key definition.

\begin{definition}[Well-structured transition system with strong compatibility]\label{def:wsts}
A {\em well-structured transition system with strong compatibility}
is a transition system
$TS = (S,\rightarrow)$, equipped with a quasi-order $\leq$ on
$S$, such that the two following conditions hold:
\begin{enumerate}
\item $\leq$ is a well-quasi-order;
\item $\leq$ is strongly (upward) compatible with $\rightarrow$, that is, for all $s_1\leq t_1$ and all transitions $s_1\rightarrow s_2$,
there exists a state $t_2$ such that $t_1\rightarrow t_2$ and
$s_2\leq t_2$ holds.
\end{enumerate}
\end{definition}

The next theorem is a special case of a result in \cite{FS01} and
will be used to obtain our decidability result.

\begin{theorem} \label{divdec}
Let $TS = (S,\rightarrow,\leq)$ be a finitely branching,
well-structured transition system
with strong compatibility, decidable $\leq$ and computable $Succ(s)$ for $s\in S$.
Then the existence of an infinite computation starting from
a state $s\in S$ is decidable.
\end{theorem}

\medskip
{\bf Decidability of divergence.}
Consider a given goal $G$ and a (CHR) program $P$ and consider the  transition system
$T= (\Conf,\stackrel{\omega_o}{\rightarrow_P})$ defined in Section \ref{subsection:original_operational_semantics}.
Obviously the number of constants and variables appearing in $G$ or in $P$ is finite. Moreover, observe that since we consider range-restricted programs, the application of the transitions $\stackrel{\omega_o}{\rightarrow_P}$ does not introduce new variables in the computations. In fact, even though rules are renamed (in order to avoid clash of variables),
the definition of the Apply rule %
(in particular the definition of $\theta$) implies that in a transition $s_1 \stackrel{\omega_o}{\rightarrow_P} s_2$
we have that $Var(s_2)\subseteq Var(s_1)$ holds. Hence an obvious inductive argument implies that no new
variables arise in computations. For this reason, given a goal $G$ and a program $P$, we can assume that the
set $\Conf$ of all the configurations uses only a finite number of constants and variables.
In %
the following
we implicitly make this assumption. %
We  define a quasi-order on configurations as follows.

\begin{definition}
Given two configurations $s_1 = \langle G_1,S_1,B_1 \rangle_i$ and $s_2 = \langle G_2,S_2,B_2 \rangle_j$ %
we say that $s_1 \leq s_2$ if
\begin{itemize}
 \item
  for every constraint $c \in G_1$ $|\{c \in G_1\}| \leq |\{c \in G_2\}|$
 \item
  for every constraint $c \in \{d \ . \ d\#i \in S_1\}$ $|\{i \ . \ c\#i \in S_1\}| \leq |\{i \ . \ c\#i \in S_2\}|$
 \item
  $B_1$ is logically equivalent to $B_2$
\end{itemize}
\end{definition}

The next Lemma, with proof in \cite{tech_report}, states the relevant property of $\leq$.

\begin{lemma}\label{lemma:wqo}
$\leq$ is a well-quasi-order on $\Conf$.
\end{lemma}

 Next, in order to obtain our decidability results we have to show that the strong compatibility property holds. This is the content of the following lemma whose proof is in \cite{tech_report}.

\begin{lemma}\label{Lemma:wstssc}
Given a CHR($\HLnofunctors$) program $P$, $(\Conf, \stackrel{\omega_o}{\rightarrow_P}, \leq)$ is a well-structured transition system with strong compatibility.
\end{lemma}

Finally we have the desired result.

\begin{theorem}\label{theorem:decid1}
Given a range-restricted CHR($\HLnofunctors$) program $P$ and a goal $G$, the existence of an infinite computation for $G$ in $P$ is decidable.
\end{theorem}

\begin{proof}
First observe that, due to our assumption on range-restricted programs, $T= (\Conf,\stackrel{\omega_o}{\rightarrow_P})$ is finitely branching. In fact, as previously mentioned, the use of rule Apply can not introduce new variables (and hence new different states). The thesis follows immediately from Lemma  \ref{Lemma:wstssc} and Theorem \ref{divdec}.
\end{proof}

The previous Theorem implies that range-restricted CHR($\HLnofunctors$) is strictly less expressive than Turing Machines, in the sense that there can not exist a termination preserving encoding of Turing Machines into range-restricted CHR($\HLnofunctors$). To be more precise, we consider an encoding of a Turing Machine into a CHR language as a function $f$ which, given a machine $Z$ and an initial instantaneous description $D$ for $Z$, produces a CHR program and a goal.
This is denoted by $(P,G) = f(Z,D)$. Hence we have the following.

\begin{definition}[Termination preserving encoding]
An encoding $f$ of Turing Machines into a CHR language is termination preserving\footnote{For many authors the existence of a termination preserving encoding into a non-deterministic  language $L$ is equivalent to the Turing completeness of $L$, however there is no general agreement on this, since for others a weak termination preserving encoding suffices.}
if the following holds: the machine $Z$ starting with $D$ terminates iff the goal $G$ in the CHR program $P$ has only terminating computations, where $(P,G) = f(Z,D)$. The encoding is weak termination preserving if: the machine $Z$ starting with $D$ terminates iff the goal $G$ in the CHR program $P$ has at least one terminating computation.
\end{definition}

Since termination is undecidable for Turing Machines, we have the following immediate corollary of Theorem \ref{theorem:decid1}.

\begin{corollary}\label{theorem:decid2}
There exists
no termination preserving encoding of Turing Machines into range-restricted CHR($\HLnofunctors$).
\end{corollary}

Note that the previous result does not exclude the existence of weak encodings. For example, in \cite{BGZ04} it is showed
that the existence an infinite computation is decidable in CCS!, a variant of CCS, yet it is possible to provide a weak termination preserving encoding of Turing Machines in CCS! (essentially by adding spurious non-terminating computations).
We conjecture that such an encoding is not possible for CHR($\HLnofunctors$). Note also that previous results imply that  range-restricted CHR($\HLnofunctors$) is strictly less
expressive than CHR($\HLnofunctors$): in fact there exists  a termination preserving encoding of Turing Machines into CHR($\HLnofunctors$) \cite{DBLP:conf/iclp/Sneyers08,DBLP:conf/sofsem/GiustoGM09}.

\section{Single-headed CHR($\HLnofunctors$)}\label{sec:sh}

As mentioned in the introduction, while CHR($\HLnofunctors$) and CHR$_1$($\HLfunctors$)
are Turing complete languages \cite{DBLP:conf/iclp/Sneyers08,DBLP:conf/sofsem/GiustoGM09},
the question of the expressive power of CHR$_1$($\HLnofunctors$) is open. Here we answer to
this question by proving that the existence of a terminating computation is decidable for this language,
thus showing that CHR$_1$($\HLnofunctors$) is less expressive than Turing machines. Throughout this section, we assume that the abstract semantics $\omega_o$ is considered (however see the discussion at the end for an extension to the case of $\omega_t$).
The proof we provide is a direct one, since neither well-structured transition systems nor reduction to Petri Nets can be used here (see the introduction).

\subsection{Some %
preparatory results}

We introduce here two more notions, namely the forest associated to a computation and the notion of reactive sequence, and some related results. We will need them for the main result of this section.

First, we observe that it is possible to associate to the computation for an atomic goal $G$ in a program $P$ a tree where, intuitively, nodes are labeled by  constraints (recall that these are atomic formulae), the root is $G$ and every child node is obtained from the parent node by firing a rule in the program $P$. This notion is defined precisely in the following, where we generalize it to the case of a generic (non atomic) goal, where for each CHR constraint in the goal we have a tree. Thus we obtain a {\em forest}  $F_\delta = (V,E)$ associated to a  computation  $\delta$, where $V$ contains a node for each repetition of identified CHR constraints in $\delta$. Before defining the forest we need the concept of {\em repetition} of an identified CHR atom in a computation.

\begin{definition}[Repetition]
Let  $P$ be a CHR program  and let
$\delta$ be a computation  in $P$. We say that an occurrence of an identified  CHR constraint $h\#l$ in $\delta$ is the $i$-th repetition of $h\#l$, denoted by $h\#l^i$, if it is preceded  in $\delta$ by $i$ $\bf{Apply}$ transitions of propagation rules whose heads match the atom $h\#l$.
We also define \[r(\delta, h\#l)=max\{i \mid \mbox{ there exists a $i$-th repetition of $h\#l$ in $\delta$}\}\]
\end{definition}

\begin{definition}[Forest]\label{def:forest}
 Let $\delta$ be a terminating computation for a goal in a  CHR$_1$($\HLnofunctors$) program. The forest associated to $\delta$, denoted by $F_\delta = (V,E)$ is defined as follows.  $V$ contains nodes labeled either by repetitions of identified CHR constraints in $\delta$ or by $\Box$. $E$ is the set of edges.
 The labeling and the edges in $E$ are defined %
 as follows:

  \noindent
 (a)    For each CHR constraint $k$ which occurs in the first configuration of $\delta$ there exists a tree in  $F_\delta = (V,E)$,  whose root is labeled by a repetition $k\#{l}^{0}$, where $k\#{l}$ is the identified CHR constraint associated to $k$ in $\delta$.

 \noindent
 (b)   If  $n$ is a node in $F_\delta = (V,E)$ labeled by $k\#{l}^{i}$ and the rule ${\it r}\ @ h \odot g \mid C, k_1,\ldots, k_m$ is used in $\delta$ to rewrite the repetition $h\#l^i$,
    where $\odot \in \{\Longleftrightarrow, \Longrightarrow\}$, the $k_i's$ are  CHR constraints while $C$ contains built-ins, then we have two cases:

\noindent
\begin{enumerate}
 \item
 If $\odot$ is $\Longrightarrow$ then $n$ has $m+1$ sons, labeled by $k_j\#{l_j}^{0}$, for $j \in [1,m]$, and by $h\#l^{i+1}$, where the $k_j\#{l_j}^{0}$ are the repetitions generated by the application of the rule $r$ to $h\#l^i$ in $\delta$.
 \item
 If $\odot$ is $\Longleftrightarrow$ then:
       \begin{itemize}
        \item
          if $m>0$ then  $n$ has $m$ sons, labeled by $k_j\#{l_j}^{0}$, for $j \in [1,m]$, where $k_j\#{l_j}^{0}$ are the repetitions generated by the application of the rule $r$ to $h\#l^i$ in $\delta$.
         \item
          if $m=0$ then $n$ has $1$ son, labeled by $\Box$.
        \end{itemize}
      \end{enumerate}

 \end{definition}

Note that, according to the previous definition, nodes which are not leaves are labeled by repetitions of identified constraints $k\#{l}^{i}$, where either $i< r(\delta, h\#l)$ or $h\#{l}$ does not occur in the last configuration of $\delta$. On the other hand, the leaves of the trees in $F_\delta$ are labeled either by $\Box$  or by the repetitions which do not satisfy the condition above. An example can help to understand this crucial definition.

\begin{example}
\label{label:example_forest}
Let us consider the following program $P$:

\begin{verbatim}
        r1 @ c(X,Y) <=> c(X,Y),c(X,Y)
        r2 @ c(X,Y) <=> X = 0
        r3 @ c(0,Y) ==> Y = 0
        r4 @ c(0,0) <=> true
\end{verbatim}

\noindent
There exists a terminating computation  $\delta$ for the goal $c(X,Y)$ in the program $P$, which uses the clauses ${\tt r1},{\tt r2},{\tt r3},{\tt r4}$ in that order and whose associated forest $F_\delta$ is the following tree:

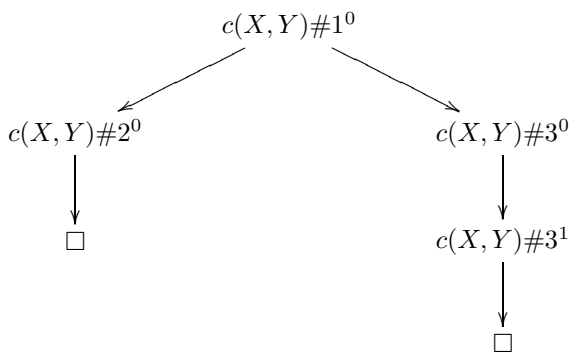
\begin{center}
\xymatrix{
& c(X,Y)\#1^0 \ar[dr] \ar[dl]\\
c(X,Y)\#2^0 \ar[d] && c(X,Y)\#3^0 \ar[d]\\
\Box && c(X,Y)\#3^1 \ar[d]\\
&& \Box
}
\end{center}

Note that the left branch corresponds to the termination obtained by using  rule {\tt r2},
hence the superscript is not incremented. On the other hand, in the right branch the superscript $^0$ at
the second level becomes $^1$ at the third level. This indicates that a propagation rule (rule {\tt r3}) has been applied.
\end{example}

Given a forest  $F_\delta$, we write $T_\delta(n)$ to denote the subtree of $F_\delta$ rooted in the node $n$. Moreover, we identify a node with its label and we omit the specification of the repetition, when not needed.
The following definition introduces some further terminology that we will need later.

\begin{definition}
\begin{itemize}
\item Given a forest  $F_\delta$, a path from a root of a tree in the forest to a leaf  is called a {\em single constraint computation}, or {\em sc-computation } for short.

  \item Two repetitions $h\#l^i$ and $k\#m^j$ of identified CHR constraints are called  r-equal, indicated by $h\#l^i == k\#m^j$, iff there exists a renaming $\rho$ such that $h= k\rho.$
  \item a sc-computation  $\sigma$ is $p$-repetitive if $p = \max_{h\#l^i \in \sigma} | \{ k\#m^j \in \sigma \mid \ h\#l^i == k\#m^j  \} |.$
  \item The degree of a $p$-repetitive sc-computation  $\sigma$, denoted by $dg(\sigma)$ is the cardinality of the set $P\_REP$ which is defined as  the maximal set having the following properties:
  \begin{itemize}
   \item
    contains a repetition $h\#l^i$ in $\sigma$ iff $p=| \{ k\#m^j \in \sigma \mid \ h\#l^i == k\#m^j  \} |$
   \item
    if $h\#l^i$ is in $P\_REP$ then $P\_REP$ does not contain a repetition $k\#m^j$ s.t. $h\#l^i == k\#m^j$
  \end{itemize}

  \item A forest $F_\delta$ is $l$-repetitive if one of its sc-computation $\sigma$ is $l$-repetitive and there is no $l'$-repetitive sc-computation $\sigma'$ in $F_\delta$ with $l' > l$.
  \item  The degree $dg(F_\delta)$ of an $l$-repetitive forest $F_\delta$ is defined as
   \[dg(F_\delta)= \sum_\sigma \{dg(\sigma) \mid  \ \sigma \mbox{ is an $l$-repetitive sc-computation  in } F_\delta\}.\]
\end{itemize}
\end{definition}

After the forest, the second main notion that we need to introduce is that one of reactive sequence\footnote{This notion is similar to that one used in the (trace) semantics of concurrent languages, see, for example, \cite{BP90a,BGM00} for the case of concurrent constraint programming. The name comes from this field.}.

Given a computation  $\delta$, we associate
to each (repetition of an) occurrence of an identified CHR atom $k\#{l}$ in $\delta$ a, so called, reactive sequence of the form $\langle c_1 , d_1\rangle \ldots \langle c_n , d_n\rangle$,
where, for any $i \in [1,n]$,  $c_i, \ d_i $ are built-in constraints.

Intuitively each pair $\langle c_i , d_i\rangle$ of built-in constraints  represents  all the $\bf{Apply}$ transition steps, in the computation $\delta$, which are used to rewrite  the considered occurrence of the identified CHR atom
$k\#{l}$  and the  identified atoms derived from it. The constraint $c_i$ represents the input for this sequence of  $\bf{Apply}$ computation  steps, while $d_i$ represents the output of such a sequence. Hence one can also read such a pair as follows: the identified CHR constraint $k\#{l}$, in $\delta$, can transform the built-in store from $c_i$ to
$d_i$. Different pairs $\langle c_i , d_i\rangle$ and $\langle c_j , d_j\rangle$ in the reactive sequence correspond to different sequences of $\bf{Apply}$ transition steps.
This intuitive notion is further clarified later (Definition \ref{def:seqrep}), when we will consider a reactive sequence associated to a repetition of an identified CHR atom.

Since in CHR computations the built-in store evolves monotonically, i.e. once a constraint is added it can not be retracted, it is natural to assume that reactive sequences are monotonically increasing. So in the following we
will assume that, for each reactive sequence $\langle c_1 , d_1\rangle \ldots \langle c_n , d_n\rangle$, the following condition holds: $CT \models d_j \rightarrow c_j$ and $CT \models c_{i+1} \rightarrow d_i$ for $j \in [1,n]$, $i \in [1,n-1]$.
Moreover, we denote the empty sequence by $\varepsilon$. Next, we define the strictly increasing reactive sequences w.r.t. a set of variables $X$.

\begin{definition}[Strictly increasing sequence]
Given a reactive sequence $s=\langle c_1, d_1\rangle \cdots \langle c_n, d_n \rangle$, with $n \geq 0$ and a set of variables $X$, we say that $s$ is strictly increasing with respect to $X$ if the following holds for any $j \in [1,n]$, $i \in [1,n-1]$
\begin{itemize}
  \item $Fv(c_j, d_j) \subseteq X$,
  \item $CT \models d_i \not \rightarrow c_{i+1}$ and $CT \models c_i \not \rightarrow d_{i}$.
\end{itemize}
\end{definition}

Given a generic reactive sequence $s=\langle c_1, d_1\rangle \cdots \langle c_n, d_n \rangle$ and a set of variables $X$, we can construct a
new, strictly increasing sequence $\eta(s,X)$ with respect to a set of variables $X$ as follows. First  the operator $\eta$ restricts all the constraints in $s$ to the variables in $X$ (by considering the existential closure with the exception of the variables in $X$). Then  $\eta$ removes from the sequence all the stuttering steps (namely the pairs of constraints $\langle c, d \rangle$, such that $CT \models c \leftrightarrow d$) except the last.  Finally, in the sequence produced by the two previous steps, if there exists a pair of consecutive elements $\langle c_l, d_l\rangle \langle c_{l+1}, d_{l+1} \rangle$ which are ``connected'', in the sense that $c_{l+1}$ does not provide more information than $d_l$, then such a pair is ``fused'' in (i.e., replaced by) the unique element $\langle c_l,  d_{l+1} \rangle$ (and this is repeated inductively for the new pairs). This is made precise by the following definition.

\begin{definition}[Operator $\eta$]\label{def:eta}
Let $s=\langle c_1, d_1\rangle \cdots \langle c_n, d_n \rangle$ be a sequence of pairs of built-in stores and  let $X$ be a set of variables.  The sequence $\eta(s,X)$ is the obtained as follows:
\begin{description}
  \item[1] First we define      $s'=\langle c'_1, d'_1\rangle \cdots \langle c'_n, d'_n \rangle$, where for $j \in [1,n]$
      $c'_j= \exists _{-X}c_j$ and $d'_j= \exists _{-X}d_j$.
      \item[2] Then we define $s''$ as the sequence obtained from $s'$ by removing each pair of the form $\langle c, d\rangle$ such that $CT \models c \leftrightarrow d$,  if such a pair is not the last one of the sequence.
   \item[3] Finally we define $\eta(s,X) = s'''$, where $s'''$ is the closure of $s''$ w.r.t. the following operation: if  $\langle c_l, d_l\rangle \langle c_{l+1}, d_{l+1} \rangle$ is a pair of consecutive elements in the sequence and $CT \models d_l  \rightarrow c_{l+1}$  holds then such a pair  is substituted  by  $\langle c_l,  d_{l+1} \rangle$.
\end{description}
\end{definition}

The following Lemma states a first useful property. The proof is in \cite{tech_report}.

\begin{lemma}\label{lem:maxlungh}
Let $X$ be a finite set of variables and let $s= \langle c_1, c_2\rangle \cdots \langle c_{n-1}, c_n \rangle$ be a strictly increasing sequence  with respect to $X$. Then $n \leq |X| +2$.
\end{lemma}

Next we note that, given a set of variables $X$ the possible strictly increasing sequences w.r.t. $X$ are finite (up to logical equivalence on constraints), if the set of the constants is finite. This is the content of the following lemma, whose proof is in \cite{tech_report}. Here and in the following, with a slight abuse of notation, given two reactive sequences $s= \langle c_1, d_1\rangle \cdots \langle c_n, d_n \rangle$ and $s'= \langle c'_1, d'_1\rangle \cdots \langle c'_n, d'_n \rangle$, we say that
$s$ and $s'$ are equal (up to logical equivalence) and we write $s=s'$, if for each $i\in [1, n]$
$CT \models c_i \leftrightarrow c'_{i}$ and $CT \models d_i \leftrightarrow d'_{i}$ holds.

\begin{lemma}\label{theo:nummax}
Let $Const$ be a finite set of constants  and let $S$ be a finite set of variables such that $u=|Const|$ and $w=|S|$. The set of sequences $s$ which are strictly increasing with respect to $S$ (up to logical equivalence) is finite and has cardinality at the most
\[\frac{2^{w(u+w)(w+3)} -1}{2^{w(u+w)} -1}.\]
\end{lemma}

Finally, we show how reactive sequences can be obtained from a forest associated to a computation. First we need to define the reactive sequence associated to a repetition of an identified CHR atom in a computation. In this definition we use the operator $\eta$ introduced in Definition \ref{def:eta}.

\begin{definition}\label{def:seqrep}
Let $\delta$ be a computation for a CHR$_1$($\HLnofunctors$) program, $h\#l^j$ be a repetition of an identified CHR atom in $\delta$ and $r_1, \dots, r_n$ the sequence of the $\bf{Apply}$ transition in $\delta$ that rewrite $h\#l^j$ and all the repetitions derived from it.
If $s \stackrel{r_i}{\rightarrow_P} s'$ let $pair(r_i)$ be the pair $(\bigwedge B_1, \bigwedge B_2)$ where $B_1$ and $B_2$ are all the built-ins in $s$ and $s'$.
We will denote with $seq(h\#l^j,\delta)$ the sequence $\eta( pair(r_1) \ldots pair(r_n), Fv(h))$
\end{definition}

Finally we define the function $S_{F_\delta}$ which, given a node $n$ in a forest associated to a computation $\delta$ (see Definition \ref{def:forest}), returns a reactive sequence. Such a sequence intuitively represents the sequence of the $\bf{Apply}$ transition steps which have been used in $\delta$ to rewrite the repetition labeling $n$ and the repetitions derived from it.

\begin{definition}[Sequence associated to a node in a forest]
 Let $\delta$ be a terminating computation  and let  $F_\delta = (V,E)$ be the forest associated to it. Given a node $n$ in $F_\delta$ we define:
   \begin{itemize}
    \item if the label of $n$ is $h\#l^i$, then $S_{F_\delta}(n) = seq(h\#l^i,\delta)$;
    \item if the label of $n$ is $\Box$ then $S_{F_\delta} (n)= \varepsilon$.
   \end{itemize}
 \end{definition}

\begin{example}
Let us consider for instance the forest shown in Example \ref{label:example_forest}. The sequences associated to the nodes of this forest are:
\begin{itemize}
 \item
  $S_{F(\delta)}(c(X,Y)\#1^0) = \langle true,X = 0 \wedge Y=0 \rangle$
 \item
  $S_{F(\delta)}(c(X,Y)\#2^0) = \langle true,X = 0 \rangle$
 \item
  $S_{F(\delta)}(c(X,Y)\#3^0) = \langle X = 0, X=0 \wedge Y=0 \rangle$
 \item
  $S_{F(\delta)}(c(X,Y)\#3^1) = \langle X = 0 \wedge Y=0,X = 0 \wedge Y=0 \rangle$
\end{itemize}
\end{example}

\subsection{Decidability of termination}

We are now ready to prove the main result of the paper.
First we need the following Lemma which has some similarities to the pumping lemma
of regular and context free grammars. Indeed, if the derivation is seen as a forest,
this lemma allows us to compress a tree if in a path of the tree there are two r-equal
constraints with an equal (up to renaming) sequence.
The %
lemma is proved
in \cite{tech_report}.

Here and in the following given a node $n$ in a forest $F$ we  denote by $A_F(n)$ the label associated to $n$.

\begin{lemma} \label{lem:substitute}
Let $\delta$  be a terminating computation for the goal $G$ in the CHR$_1$($\HLnofunctors$)  program $P$. Assume that $F_\delta$ is $l$-repetitive with  $p=dg(F_\delta)$ and assume that there exists an $l$-repetitive sc-computation  $\sigma$ of $F_\delta$ and a repetition $k\#l^i \in \sigma$  such that $l = | \{h\#n^{j}\in \sigma \mid \ h\#n^{j} == k\#l^i \}|$.\\
Moreover assume that there exist two distinct nodes $n$ and $n'$ in $\sigma$ such that
$n'$ is a node in $T_\delta(n)$, $A_{F_\delta}(n)= k\#l^i$,  $A_{F_\delta}(n')= k'\#{l'}^{i'}$ and $\rho$ is a renaming such that $S_{F_\delta}(n)= S_{F_\delta}(n')\rho$ and $k = k'\rho$.

Then there exists a terminating computation  $\delta'$ for the goal $G$ in the program $P$, such that
either
$F_{\delta'}$ is $l'$-repetitive with $l'<l$, or $F_{\delta'}$ is $l$-repetitive and
$dg(F_\delta')<p$.
\end{lemma}

Finally we obtain the following result, which is the main result of this paper.

\begin{theorem}[Decidability of termination]
\label{theo:decidability_CHR_1}
Let $P$ be a CHR$_1$($\HLnofunctors$) program an let $G$ be a goal. Let  $u$ be the number of distinct constants used in $P$ and in $G$ and let $w$ be the maximal arity of the CHR constraints which occur in $P$ and in $G$.

$G$ has a terminating computation  in $P$ if and only if there exists a terminating computation  $\delta$ for $G$ in $P$ s.t. $F_\delta$ is $m$-repetitive and $m \leq \frac{2^{w(u+w)(w+3)} -1}{2^{w(u+w)} -1}=L.$
\end{theorem}

\begin{proof}
We prove only that if $G$ has a terminating computation  in $P$ then there exists a terminating computation  $\delta$ for $G$ in $P$ s.t. $F_\delta$ is $m$-repetitive and $m \leq L$. The proof of the converse is straightforward and hence it is omitted.

The proof is by contradiction.
Assume $G$ has a terminating computation  $\delta$ in $P$  s.t. $F_\delta$ is $m$-repetitive,
 $m > L$ and there is no  terminating computation  $\delta'$ for $G$ in $P$ such that $F_{\delta'}$ is $m'$-repetitive and $m' < m$.
Moreover, without loss of generality, we can assume that the degree of $F_\delta$ is minimal, namely there is no terminating computation  $\delta'$ for $G$ in $P$ such that $F_{\delta'}$ is $m$-repetitive and
 $dg(F_{\delta'}) < dg (F_{\delta})$.

  Let $\sigma$ be  a $m$-repetitive sc-computation  in $F_{\delta}$. By definition, there exist $m$ repetitions of identified CHR constraints  $k_1\#{l_1}^{i_1},...,k_r\#{l_m}^{i_m}$ in $\sigma$, which are $r$-equal. Therefore there exist renamings $\rho_{s,t}$ such that $k_s= k_t \rho_{s,t}$ for each $s,t \in [1,m]$.

By Lemma \ref{theo:nummax} for each CHR constraint $k$ which occurs in $P$ or in $G$, the set of sequences $s$ which are strictly increasing with respect to $Fv(k)$ (up to logical equivalence) is finite and has cardinality at the most $L$. Then there are two distinct nodes $n$ and $n'$ in $\sigma$ and there exist $s,t \in [1,m]$ such that
$A(n)=k_s\#{l_s}^{i_s}$ and $A(n')=k_t\#{l_t}^{i_t}$ and
$S_{F_\delta}(n)= S_{F_\delta}(n')\rho_{s,t}$.
  Then we have a contradiction, since by Lemma \ref{lem:substitute} this implies that there exists a  terminating computation  $\delta'$ for $G$ in $P$ s.t. either $F_{\delta'}$ is $m'$-repetitive with $m' < m$ or $F_{\delta'}$ is $m$-repetitive and $dg(F_{\delta'}) < dg (F_{\delta})$ and then the thesis.
\end{proof}

As an immediate corollary of the previous theorem we have that the existence of a terminating computation for a goal $G$ in a CHR$_1$($\HLnofunctors$)  program $P$ is decidable. Then we have also the following result, which is stronger than Corollary \ref{theorem:decid2} since here weak encodings are considered.

\begin{corollary}\label{theorem:decid3}
There %
is no weak termination preserving encoding of Turing Machines into  CHR$_1$($\HLnofunctors$).
\end{corollary}

As mentioned at the beginning of this section, the previous result is obtained when considering the abstract semantics $\omega_o$. However it holds also when considering the theoretical semantics $\omega_t$. In fact Lemma \ref{lem:substitute} holds if we require that two r-equal constraints have the same sequence and have fired the same propagation rules. Since the propagation rules are finite Theorem \ref{theo:decidability_CHR_1} is still valid if $m \leq 2^r \cdot \frac{2^{w(u+w)(w+3)} -1}{2^{w(u+w)} -1} $ where $r$ is the number of propagation rules.

\section{Conclusions}\label{sec:conclusions}

We have shown two decidability results for two fragments of CHR($\HLnofunctors$),
the CHR language defined over a signature which does not allow function symbols.
The first result, in Section \ref{sec:rr}, assumes the abstract operational semantics,
while the second one, in Section \ref{sec:sh}, holds for both semantics
(abstract and theoretical).
These results are not immediate. Indeed, CHR($\HLnofunctors$), without further restrictions and with any of the
two semantics, is a Turing complete language  \cite{DBLP:conf/iclp/Sneyers08,DBLP:conf/sofsem/GiustoGM09}.
It remains quite expressive also with our restrictions:
for example, CHR$_1$($\HLnofunctors$), the second fragment that we have considered,
allows an infinite number of different states, hence, for example, it can not be
translated to Petri Nets.

These results imply that range-restricted CHR($\HLnofunctors$) and CHR$_1$($\HLnofunctors$), the two considered fragments, are strictly less expressive than
Turing Machines (and therefore than CHR($\HLnofunctors$)). Also, it seems that
range-restricted CHR($\HLnofunctors$) is more expressive that CHR$_1$($\HLnofunctors$),
since the decidability result for the second language is stronger.
However, a direct result in this sense is left for future work.
Also, we leave to future work to establish a decidability result for
range-restricted CHR($\HLnofunctors$) under an operational semantics which includes a propagation history.
This is not easy, since in this case it appears difficult to apply the theory of well-structured transition systems (the well-quasi-order we have defined does not work).

Several papers have considered the expressive power of CHR in the last few years.
In particular, \citeN{DBLP:conf/iclp/Sneyers08} showed that a further restriction of
CHR$_1$($\HLnofunctors$), which does not allow built-ins in the body of rules
(and which therefore does not allow unification of terms) is not Turing complete.
This result is obtained by translating
 CHR$_1$($\HLnofunctors$) programs (without unification) into propositional CHR and
 using the encoding of propositional CHR intro Petri Nets provided in \cite{betz_petri_nets_chr07}.
The translation to propositional CHR is not possible for the language (with unification)  CHR$_1$($\HLnofunctors$)
that we consider.
\citeN{betz_petri_nets_chr07} also provides a translation of range-restricted CHR($\HLnofunctors$) to Petri nets.
However in this translation, differently from our case, it is also assumed that no unification built-in
can be used in the rules, and only ground goals are considered. Related to this paper is also \cite{DBLP:conf/sofsem/GiustoGM09}, where it is shown that CHR($\HLfunctors$) is Turing complete and that restricting to single-headed rules decreases the computational power of CHR. However, these results are based on the theory of language embedding, developed in the field of concurrency theory to compare Turing complete languages, hence they do not establish any decidability result. Another related study  is \cite{sney_schr_demoen_chr_complexity_08}, where the authors show that it is possible to implement any algorithm in CHR in an efficient way, i.e. with
the best known time and space complexity.   Earlier works by Fr\"uhwirth \cite{DBLP:journals/tplp/FruhwirthA01,DBLP:conf/kr/Fruhwirth02} studied the time complexity of simplification rules for naive implementations of CHR.
In this approach an upper bound on the derivation length,
combined with a worst-case estimate of (the number and cost of) rule application attempts, allows to obtain an upper bound of the time complexity. The aim of  all these works is clearly  different from ours.

\begin{table}
\begin{center}
\renewcommand{\arraystretch}{1.2}
\begin{oldtabular}{p{0.35\textwidth}|c|c|c}
{\it Host language $X$ }  & {\it Operational semantics } & $k=1$ & $k>1$ \\
                                            \cline{3-4}
\cline{1-4}
P (propositional)
                               & abstract   & No & No \\
\cline{1-4}
\multirow{2}{0.35\textwidth}{range-restricted C (constants)
(cf.~Section~\ref{sec:rr})}
& abstract   & No & {\bf No} \\
& & & \\
\cline{1-4}
C (constants), without =
                & any        & No & Yes \\
\cline{1-4}
C (constants)
(cf.~Section~\ref{sec:sh})
                & any                     & {\bf No} & Yes \\
\cline{1-4}
F (functors)           & any       & Yes & Yes \\
\end{oldtabular}
\end{center}
\caption{Termination preserving encoding of Turing Machines into CHR$_k$($X$)}
\label{tbl:overview}
\end{table}

A summary  of the existing results concerning the computational power of several dialects of CHR
is shown in Table \ref{tbl:overview}. In this table, ``no'' and ``yes'' refer to
the existence of a termination preserving encoding of Turing Machines into the considered language, while  ``any'' means theoretical or abstract.
The new results shown in this paper are indicated in a bold font.

\subsubsection*{Acknowledgments.}
We would like to thank the reviewers for their precise and helpful comments. This research was partially supported by the MIUR PRIN 20089M932N project: "Innovative and multi-disciplinary approaches for constraint and preference reasoning".

\end{document}